\PassOptionsToPackage{unicode}{hyperref}
\PassOptionsToPackage{hyphens}{url}
\documentclass[
]{article}
\usepackage{amsmath,amssymb}
\usepackage{iftex}
\ifPDFTeX
  \usepackage[T1]{fontenc}
  \usepackage[utf8]{inputenc}
  \usepackage{textcomp} 
\else 
  \usepackage{unicode-math} 
  \defaultfontfeatures{Scale=MatchLowercase}
  \defaultfontfeatures[\rmfamily]{Ligatures=TeX,Scale=1}
\fi
\usepackage{lmodern}
\ifPDFTeX\else
\fi
\IfFileExists{upquote.sty}{\usepackage{upquote}}{}
\IfFileExists{microtype.sty}{
  \usepackage[]{microtype}
  \UseMicrotypeSet[protrusion]{basicmath} 
}{}
\makeatletter
\@ifundefined{KOMAClassName}{
  \IfFileExists{parskip.sty}{%
    \usepackage{parskip}
  }{
    \setlength{\parindent}{0pt}
    \setlength{\parskip}{6pt plus 2pt minus 1pt}}
}{
  \KOMAoptions{parskip=half}}
\makeatother
\usepackage{xcolor}
\usepackage{color}
\usepackage{fancyvrb}

\DefineVerbatimEnvironment{Highlighting}{Verbatim}{commandchars=\\\{\}}
\usepackage{framed}
\definecolor{shadecolor}{RGB}{248,248,248}
\newenvironment{Shaded}{\begin{snugshade}}{\end{snugshade}}

\newcommand{\BuiltInTok}[1]{#1}

\newcommand{\ControlFlowTok}[1]{\textcolor[rgb]{0.13,0.29,0.53}{\textbf{#1}}}
\newcommand{\DataTypeTok}[1]{\textcolor[rgb]{0.13,0.29,0.53}{#1}}

\newcommand{\FloatTok}[1]{\textcolor[rgb]{0.00,0.00,0.81}{#1}}
\newcommand{\FunctionTok}[1]{\textcolor[rgb]{0.13,0.29,0.53}{\textbf{#1}}}
\newcommand{\ImportTok}[1]{#1}

\newcommand{\KeywordTok}[1]{\textcolor[rgb]{0.13,0.29,0.53}{\textbf{#1}}}
\newcommand{\NormalTok}[1]{#1}
\newcommand{\OperatorTok}[1]{\textcolor[rgb]{0.81,0.36,0.00}{\textbf{#1}}}

\newcommand{\StringTok}[1]{\textcolor[rgb]{0.31,0.60,0.02}{#1}}

\usepackage{graphicx}
\makeatletter
\def\maxwidth{\ifdim\Gin@nat@width>\linewidth\linewidth\else\Gin@nat@width\fi}
\def\maxheight{\ifdim\Gin@nat@height>\textheight\textheight\else\Gin@nat@height\fi}
\makeatother
\setkeys{Gin}{width=\maxwidth,height=\maxheight,keepaspectratio}
\makeatletter
\def\fps@figure{htbp}
\makeatother
\ifLuaTeX
  \usepackage{selnolig}  
\fi
\IfFileExists{bookmark.sty}{\usepackage{bookmark}}{\usepackage{hyperref}}
\IfFileExists{xurl.sty}{\usepackage{xurl}}{} 
\hypersetup{
  pdftitle={Uniform linear extensions using few random bits},
  pdfauthor={Mark Huber},
  hidelinks,
  pdfcreator={LaTeX via pandoc}}

\title{Generating uniform linear extensions using few random bits}
\author{Mark Huber}

\usepackage{amsthm}
\newtheorem{theorem}{Theorem}
\newtheorem{lemma}{Lemma}

\begin{document}
\maketitle

\newcommand{\unifdist}{\textsf{Unif}}
\newcommand{\prob}{\mathbb{P}}
\newcommand{\mean}{\mathbb{E}}
\newcommand{\var}{\mathbb{V}}
\newcommand{\ind}{\mathbb{I}}

\begin{abstract}
A \emph{linear extension} of a partial order
\(\preceq\) over items \(A = \{ 1, 2, \ldots, n \}\) is a permutation
\(\sigma\) such that for all \(i < j\) in \(A\), it holds that
\(\neg(\sigma(j) \preceq \sigma(i))\). Consider the problem of
generating uniformly from the set of linear extensions of a partial
order. The best method currently known uses \(O(n^3 \ln(n))\) operations
and \(O(n^3 \ln(n)^2)\) iid fair random bits to generate such a
permutation. This paper presents a method that generates a
uniform linear extension
using only \(2.75 n^3 \ln(n)\) operations and \( 1.83 n^3 \ln(n) \) iid fair bits on average.
\end{abstract}

\hypertarget{introduction}{%
\section{Introduction}\label{introduction}}

Consider a \emph{partial order} \(\preceq\) on a set of items
\(S = \{1, 2, \ldots, n \}\). Such a partial order \(\preceq\) is a
binary relation with the following three properties.

\begin{enumerate}
\def\labelenumi{\arabic{enumi})}
\item
  \textbf{Reflexive} For all \(i \in S\), \(i \preceq i\).
\item
  \textbf{Antisymmetric} For all \(i, j \in S\), if \(i \preceq j\) and
  \(j \preceq i\) then \(i = j\).
\item
  \textbf{Transitive} For all \(i, j, k \in S\), if \(i \preceq j\) and
  \(j \preceq k\), then \(i \preceq k\).
\end{enumerate}

A permutation of the items \(\{1, \ldots, n \}\) can be viewed as
putting the items in a sorted order. For instance write
\(\sigma = (4, 1, 3, 2)\) if position 1 contains item 4, position 2
contains item 1, and so on.  Say that item 1 comes before item 3 and item 2 in this example.

A permutation is a \emph{linear extension} if when item \( i \) comes
before item \( j \), it cannot hold that \( j \preceq i \).  That is, a linear extension of \( \preceq \) is a permutation \( \sigma \) such 
that for all \(i < j\), \(\neg(\sigma(j) \preceq \sigma(i))\).

Linear extensions have been extensively studied \cite{kahn1984balancing,brightwell1991counting,bubley1999faster,karzanov1991conductance} and appear in a wide variety of contexts, such as scheduling~\cite{cohen2001labeling}, decision theory~\cite{fishburn1985interval}, election modeling~\cite{kenig2019approximate}, and volume computation~\cite{stanley1986two}.

Let \( \Omega \) denote the set of linear extensions of \(S\) with respect
to \( \preceq \).   Consider the problem of generating uniformly from the
set of linear extensions.  Using techniques such as those pioneered by Jerrum, Valiant, and Vazirani~\cite{jerrum1986random} and the Tootsie Pop Algorithm~\cite{huber2018c}, it is possible to use approximate or exact generation techniques to build a fully polynomial approximation scheme for the associated counting problem.  Counting the number of linear extensions of a poset is known to be \( \# \)P-complete~\cite{brightwell1991counting}.

In~\cite{huber2006b}, it was shown how to accomplish the task of generating exactly uniformly from the set of linear extensions
in expected time \(O(n^3 \ln(n))\) using a non-Markovian coupling.
The expected number of iid uniform bits used by this algorithm was
\(\Theta(n^3 \ln(n)^2)\).

For certain classes of problem, faster algorithms are known.  For instance, a nontrivial poset is \emph{height}-2 if there does not exist a chain of distinct items \( i, j, k \) where \( i \preceq j \preceq k \).  For height-2 posets, it is possible to generate a sample using (on average) \( O(n \ln(n)) \) operations~\cite{huber2014near}.  Dittmer and Pak~\cite{dittmer2018counting} showed that the associated counting problem remains \( \# \)P-complete in this restricted class of posets.

While other approaches to generation have been studied~\cite{combarro2019minimals, talvitie2018counting, talvitie2024approximate}, the \( O(n^3 \ln(n)) \) method remains the best known polynomial bound.  This work improves upon the previous algorithm in several ways.

\begin{enumerate}
\def\labelenumi{\arabic{enumi})}
\item
  The old algorithm used one fair random bit and one uniform draw from
  \(\{1, \ldots, n - 1\}\) at each operation. A draw from \( \{1, \ldots, n - 1\} \) takes 
  \(\Theta(\ln(n))\) iid random bits. The new algorithm uses the same order
  of operations, but now each operation takes at most one fair random bit.
\item
  The constant in the expected number of operations is reduced from
  \(16 / \pi \approx 5.093\) to \(2.75\).
\end{enumerate}

The result is a faster algorithm whose expected number of operations
remains \(O(n^3 \ln(n))\), but now the number of iid fair random bits
used equals the number of operations.

\begin{theorem}
There exists an algorithm for drawing uniformly from
the linear extensions of a poset that uses on average at most
\(n(n^2 - n + 3) \lceil \log_2(n) \rceil \leq 2.75 n^3 \ln(n)\)
operations, each of which evaluates \( i \preceq j \) once for two adjacent values of a permutation, and possibly swaps them.  The algorithm also  uses
on average \( 1.83 n^3 \ln(n) \) iid fair random bits.
\end{theorem}

The rest of the paper is organized as follows. The following section
discusses the original chain used in~\cite{huber2006b} and shows how
to use a faster chain to accomplish the same goal. Section~\ref{SEC:analysis} analyzes this new method and
gives the running time bound.  The final section implements the algorithm in Julia.

\hypertarget{adjacent-transposition-markov-chains-for-linear-extension}{%
\section{Adjacent transposition Markov chains for linear
extension}\label{adjacent-transposition-markov-chains-for-linear-extension}}

The Markov chain used in~\cite{huber2006b} is a simple Gibbs chain
that uses transpositions of items in adjacent positions.
First a position \( i \in \{1, \ldots, n - 1 \} \) is chosen uniformly
at random.  Next, if \( \sigma(i) \preceq \sigma(i + 1) \) then the
state remains the same.  Otherwise, items \( \sigma(i) \) and \( \sigma(i + 1) \) are swapped with probability \( 1 / 2 \).  This chain was analyzed
in Karzanov and Khachiyan (1991)~\cite{karzanov1991conductance}
and shown to mix in \(O(n^3 \ln(n))\) time.

The following is a representation of this chain.

\(\textbf{adj\_transpos\_chain\_step} (\sigma, i, c)\)

\begin{enumerate}
\def\labelenumi{\arabic{enumi})}
\item
  If \(c = 1\) and \(\neg(\sigma(i) \preceq \sigma(i + 1))\), set
  \((\sigma(i), \sigma(i + 1)) \leftarrow (\sigma(i + 1), \sigma(i))\).
\item
  Return \(\sigma\).
\end{enumerate}

Since Gibbs chains are reversible, the following lemma follows directly.

\begin{lemma} Suppose \(\sigma\) is uniform over \(\Omega\), \(i\) is
uniform over \(\{1, \ldots, n - 1\}\), and \(c\) is uniform over
\(\{0, 1\}\). Then the output
\textbf{adj\_transpos\_chain\_step}\((\sigma, i, b)\) is also
uniform over \(\Omega\).
\end{lemma}

\hypertarget{the-adjacent-position-transposition-bounding-chain}{%
\subsection{The adjacent position transposition bounding
chain}\label{the-adjacent-position-transposition-bounding-chain}}

A \emph{bounding chain} has state space consisting of \(2^\Omega\). That
is, each state of the bounding chain is a subset of possible states in
the original chain.

In this case, the representation of the set of states is given by a
vector of length \(n\) whose components are drawn from
\(\{*\} \cup \{1, 2, \ldots, n \}\) together with an integer that
indicates the number of active states. Extend the partial order so that
for all \(i \in \{1, 2, \ldots, n\}\), the statement \(* \preceq i\) is false.

For simplicity, assume without loss of generality that the identity
permutation \((1, 2, \ldots, n)\) is a linear extension. If not, simply
sort the items in \(O(n \ln(n))\) time using the partial order and
relabel the items so that this is true. This will greatly simplify the
description of the bounding chain.

Before giving the formal description, consider
some examples. Suppose the bounding state
is \(((*, *, *, *, 1), 1)\). This indicates that only item 1 is active,
and that it lies anywhere in positions 1 through 5.

Now suppose the bounding state is \(((*, 2, *, 1, 3), 3)\). This
indicates that item 2 is in position 1 or 2, item 1 is somewhere in
position 1 through 4, item 3 is somewhere in position 1 through 5, and
only these 3 items are active.

\textbf{Definition} Say that bounding chain state \[
(r, k) \in (\{*, 1, \ldots, n\}^n, \{1, 2, \ldots, n\})
\] \emph{bounds} permutation \(\sigma\) if \[
(\forall a \leq k)(\forall i, j \in \{1, \ldots, n\})(((\sigma(i) = a) \wedge (r(j) = a)) \rightarrow (i \leq j)).
\]

The following invariants will be maintained for this bounding state
\((r, k)\).

\begin{enumerate}
\def\labelenumi{\arabic{enumi})}
\item
  (Active items) For all \(a \leq k\), there exists \(i\) such that
  \(r(i) = a\).
\item
  (Partial order holds on active items) If \(r(i) = a \leq k\),
  \(r(j) = b \leq k\), and \(a \preceq b\), then \(i < j\).
\end{enumerate}

The following pseudocode takes one step in the bounding chain.  The only difference between this and the original chain is that if a \( * \) symbol reaches position \( n \), it is removed and replaced with \( 1 + k \).  Then the number of active items is increased by one to \( 1 + k \).

\(\textbf{bounding\_chain\_step}(y = (r, k), i, c)\)

\begin{enumerate}
\def\labelenumi{\arabic{enumi})}
\item
  If \(c = 1\) and \(\neg(r(i) \preceq r(i + 1))\), set
  \((r(i), r(i + 1)) \leftarrow (r(i + 1), r(i))\).
\item
  If \(r(n) = *\), set \(r(n) \leftarrow k + 1\), and
  \(k \leftarrow k + 1\).
\item
  Return \((r, k)\).
\end{enumerate}

\begin{lemma} The bounding chain step preserves the active item and
partial order invariants.
\end{lemma}

\begin{proof}
Note that at the beginning of the step, if all active
items \(\{1, 2, \ldots, k\}\) are present in \(r\), they will still be
present because only their positions can be switched, they cannot be
removed. If the number of active items is increased by 1, then \(r(n)\)
equals the new active item. So the active item invariant holds.

Line 1 of the bounding chain step does not allow the partial order to be violated.  If the number of active items is increased,
the assumption that the identity permutation is a linear extension means
that the partial order invariant still holds as well.
\end{proof}

\hypertarget{combining-the-adjacent-transposition-markov-chain-and-bounding-chain}{%
\subsection{Combining the adjacent transposition Markov chain and
bounding
chain}\label{combining-the-adjacent-transposition-markov-chain-and-bounding-chain}}

Suppose that \((r, k)\) bounds \(\sigma\). Consider how to make one step
in both the underlying state \(\sigma\) and the bounding chain
simultaneously in order for the bounding state to keep bounding the
underlying chain state. The combined chain can use the coin flip one way for the underlying chain and a different way for the bounding chain.

\(\textbf{simultaneous\_chain\_step}(\sigma, y = (r, k), i, c)\)

\begin{enumerate}
\def\labelenumi{\arabic{enumi})}
\item
  If \(\sigma(i) = r(i + 1)\), let \(c' \leftarrow 1 - c\), else let
  \(c' \leftarrow c\).
\item
  Return
  (\(\textbf{adj\_transpos\_chain\_step} (\sigma, i, c')\),
  \(\textbf{bounding\_chain\_step}(y, i, c)\))
\end{enumerate}

Note that for \(c\) uniform over \(\{0, 1\}\),
\(1 - c\) is also uniform over \(\{0, 1\}\), so no matter what the state
of the bounding chain relative to the underlying chain, the step taken
in the Markov chain will preserve the uniform distribution on
\(\Omega\).

The next theorem says that this combined step does what it is intended
to: if the underlying state \(\sigma\) starts out bounded by \(y\), then
after one step in the combined chain, the underlying state is still
bounded by the bounding chain state.

\begin{theorem}
If \(\sigma\) is bounded by \(y\), then for all
\(i \in \{1, \ldots, n - 1\}\) and \(c \in \{0, 1\}\), 
\( \textbf{adj\_transpos\_chain\_step}(\sigma, i, c') \) is bounded by \( \textbf{bounding\_chain\_step}(y, i, c) \).
\end{theorem}

\begin{proof}
Suppose that \(\sigma\) is bounded by
\((r, k)\).

If the number of active items is increased by 1, the new active
item \(k + 1\) is in position \(n\) so trivially bounds the position of
\(k + 1\) in the new underlying state.

There are three cases to consider.

\emph{Case 1:}  \(\sigma(i) = r(i)\) and the underlying state
swaps.  Then the goal is to show that the bounding chain must also swap so that item \( \sigma(i) \) is still bounded after moving to position \( i + 1 \).

For the underlying state to
swap requires \( c' = 1 \), and \( \sigma(i) = r(i) \) gives \( c = c' = 1 \).

Let \( a \neq r(i) \) be any active item such that \( r(i) \preceq a \).  Then there exists \( j \) such that \( r(j) = a \) and \( j > i \) by the poset respecting invariant of \( r \).  
There exists \( j' \) such that \( \sigma(j') = a \), and since \( r \) bounds \( \sigma \), \( j' \leq j \).  Since \( \sigma \) is a linear extension it holds that \( j' > i \).  Moreover \( \sigma(i + 1) \neq a \) since it was assumed that the \( \sigma \) chain swapped. So \( i + 1 < j' \leq j \).

Therefore, no active items that \( r(i) \) precedes appear as \( r(i + 1) \).  So whenever \( \sigma(i) = r(i) \) and the underlying chain swaps, the bounding chain must swap as well.  Nonactive items do not prevent the bounding chain state from swapping.

\emph{Case 2:} \(\sigma(i + 1) = r(i + 1)\) and the bounding chain swaps.  The goal here is to show that the underlying chain must have also swapped as well.

Because the bounding chain swaps \( c = 1 \) and since \( \sigma(i) \neq r(i + 1) \), \( c = c' = 1 \) so the underlying chain does attempt to swap.

If \( r(i) = \sigma(i) \), then the fact that the bounding chain swaps means that \( \neg(r(i) \preceq r(i + 1)) \), so the underlying chain swaps as well.

Otherwise, let \( a = \sigma(i) \).  If \( a \) is an active item there exists \( j \) such that \( r(j) = a \).  Because \( r \) bounds \( \sigma \) it holds that \( j > i \).  Moreover, \( j \neq i + 1 \) so \( j > i + 1 \).  Since \( r \) obeys the partial order, this means \( \neg(a \preceq r(i + 1)) \) and the underlying state must swap as well.

If \( a \) is not an active item, then \( \neg(a \preceq r(i + 1)) \).  Again in this case the underlying chain swaps as well.

\emph{Case 3:}  \( \sigma(i) = r(i + 1) \).  Here it is important to show that it cannot happen that both the underlying state and the bounding state swap, as that would move item \( \sigma(i) \) to position \( i + 1 \) and item \( r(i + 1) \) to position \( i \).  Fortunately, they cannot both swap because in this case \( c' = 1 - c \) so they cannot both be 1.
\end{proof}

\hypertarget{deterministic-versus-random-scan}{%
\subsection{Deterministic versus random
scan}\label{deterministic-versus-random-scan}}

In~\cite{huber2006b}, this combined bounding chain was applied to a
\emph{random scan} Markov chain where \(i\) is chosen uniformly from
\(\{1, \ldots, n - 1 \}\) and \(c\) is uniform over \(\{0, 1\}\).

Suppose instead that this combined bounding chain is applied to a
\emph{deterministic scan} Markov chain where a step is taken for each
\(i\) from 1 to \(n - 1\) with a new value of \(c\) flipped for each
step. The resulting update still has a stationary distribution that is
uniform over \(\Omega\).

A single step in this deterministic scan chain uses \(n - 1\)
operations, and \(n - 1\) random bits. These steps can then be used to
form a stochastic process where each component is an underlying state
together with a bounding chain state that bounds it \[
(\sigma_0, y_0), (\sigma_1, y_1), \ldots, (\sigma_t, y_t).
\] Here
\[
y_0 = (*, *, *, \ldots, *, 1)
\] and assume that \(\sigma_0\) is uniform over \(\Omega\).

Now \(\sigma_t\), because it was a stationary state run forward from a
stationary state \(\sigma_0\) for a fixed number of steps \(t\), is
still stationary.

But if \(y_t\) is a permutation, then it was unnecessary to run the
\(\sigma\) process along side it! Just running the \(y_t\) process
forward would have been enough to completely determine \(\sigma_t\). This
is the essential idea of Propp and Wilson~\cite{proppw1996} that they called
Coupling From the Past (CFTP).

What if \(y_t\) is not a permutation? Then a user can call \textbf{CFTP}
again (perhaps with a larger value of \(t\) to make it more likely to
work) to generate \(\sigma_0\), and run the same steps again to find
\((\sigma_t, y_t)\).

The Fundamental Theorem of Perfect Simulation~\cite{huber2015b} says that an
algorithm of this form will generate samples exactly from the target
distribution as long as the chance that it eventually terminates is 1.

As long as \(t > n\), there exists a sequence of moves that eliminate all
\( * \) symbols and make \(y_t\)
a permutation, so there is a positive chance of exiting at each step of
the process.

The algorithm then is as follows.

\textbf{CFTP\_linear\_extensions}\((t)\)

\begin{enumerate}
\def\labelenumi{\arabic{enumi}.}
\item
  Let \(y_0 \leftarrow (*, *, *, \ldots, *, 1).\)
\item
  For \(t'\) from 1 to \(t\)

  \begin{enumerate}
  \def\labelenumii{\alph{enumii}.}
  \item
    Let \(y \leftarrow y_{t' - 1}\).
  \item
    Draw \((c_{t', 1}, \ldots, c_{t', n - 1})\) iid uniform over
    \(\{0, 1\}\).
  \item
    For \(i\) from 1 to \(n - 1\): Let
    \(y \leftarrow \textbf{bounding\_chain\_step}(y, i, c_{t', i})\).
  \item
    Let \(y_{t'} \leftarrow y\).
  \end{enumerate}
\item
  If \(y_t\) is a permutation, return \(y_t\).
\item
  Let \(\sigma_0 \leftarrow \textbf{CFTP\_linear\_extensions}(t)\) and
  \(y_0 \leftarrow (*, *, *, \ldots, *, 1)\).
\item
  For \(t'\) from 1 to \(t\),

  \begin{enumerate}
  \def\labelenumii{\alph{enumii}.}
  \item
    Let \(\sigma \leftarrow \sigma_{t' - 1}\),
    \(y \leftarrow y_{t' - 1}\).
  \item
    For \(i\) from 1 to \(n - 1\): Let
    \((\sigma, y) \leftarrow \textbf{simultaneous\_chain\_step}(\sigma, y, i, c_{t', i})\).
  \item
    Let \(\sigma_{t'} \leftarrow \sigma\), \(y_{t'} \leftarrow y\).
  \end{enumerate}
\item
  Return \(\sigma_t\).
\end{enumerate}

\begin{lemma}
If the probability that the algorithm terminates is at least \( 1 / 2 \),
the average number of bits used is at most \( 2t \), and the average number of
step calls is at most \( 3t \).
\end{lemma}

\begin{proof}
Note that the number of recursions is geometrically distributed with mean 2.  These require \( 2t \) random bits on average.  In the final call, there are \( t \) steps taken, and in any call that calls a recursion, there at \( 2t \) steps taken.  Therefore, the expected number of steps is \( (2t)(2 - 1) + (1)(t) = 3t \).
\end{proof}

\hypertarget{analysis}{%
\section{Analysis}\label{SEC:analysis}}

Now consider how to find a value \( t \) where the probability of
termination is at least \( 1 / 2 \).

Recall that the bounding chain contains symbols \(*\) together with
numbers of item. When a particular symbol \(*\) reaches the right hand
side of the \(n\)-tuple, it disappears and is replaced by the new active item.
Let \( \tau_* \) denote the time at which this occurs. Then the following
holds.

\begin{lemma}
Let \(*\) be in the first component of the bounding chain and
\(\tau = \inf\{t:* \text{ disappears} \}\). Then \[
\mathbb{E}[\tau] \leq (n^2 - 3n + 2) / 2.
\]
\end{lemma}

\begin{proof}
Denote the position of \(*\) before the \(t\)th step of the chain by
\(p_t\). So \(p_0 = 1\). Consider what happens during the step. If
\(p_t > 1\), then when the substep at \(i = p_t - 1\) is executed, there
is a \(1 / 2\) chance of a move to the left, so \(p_{t + 1} = p_t - 1\).

Similarly, there is a \(1 / 2\) chance of not moving to the left. Then
at \(i = p_t\), there is a \(1 / 2\) chance of moving to the right. If
it does move to the right, at \(i = p_t + 1\) there is again a \(1 / 2\)
chance of moving to the right, and so on to the point where there is a
small chance that the symbol reaches the final position in a single
step.

The result is a distribution for \(p_{t + 1} - p_t\) when \(p_t > 1\):
\[
\mathbb{P}(p_{t + 1} - p_t = i) = 
 \left(\frac{1}{2}\right)^{i + 2} \mathbb{I}(i \in \{-1, 0, 1, 2, \ldots\}).
\] Here \(\mathbb{I}(q)\) is the usual indicator function that is 1 if
logical statement \(q\) is true, and 0 otherwise.

Note that this distribution assumes that \(p_t \in \{1, 2, \ldots \}\)
with no truncation at \(n\). If the \(p_t\) process moves to greater
than \(n\), then it must have been that the \(*\) symbol reached \(n\).
Conversely, \(p_t\) is at least the position of \(*\) so cannot reach
\(n\) before \(*\) does. Hence \[
\tau = \inf\{t:p_t \geq n \}
\] has the same distribution as \(\tau\) defined earlier.

This is equivalent to \([p_{t + 1} - p_t | p_t > 1] \sim G - 2\), where
\(G\) is a geometric random variable in \(\{1, 2, \ldots \}\) with mean
\(1 / 2\). So \[
\mathbb{E}[p_{t + 1} - p_t | p_t > 1] = \mathbb{E}[G - 2] = \frac{1}{1 / 2} - 2 = 2 - 2 = 0, 
\] and \begin{align*}
\mathbb{E}[p_{t + 1}^2 - p_t^2 | p_t > 1] &= -p_t^2 + 
  \sum_{i = -1}^\infty \left( \frac{1}{2}\right)^{i + 2} (p_t + i)^2 \\
  &= -p_t^2 + 
  \sum_{i = -1}^\infty \left( \frac{1}{2}\right)^{i + 2} (p_t^2 + 2 p_t i + i^2) \\
  &= 2 p_t \mathbb{E}[G - 2] + \mathbb{E}[(G - 2)^2] \\
  &= \mathbb{V}(G) + \mathbb{E}[(G - 2)]^2 \\
  &= (1 - 1/2) / (1 / 2)^2 = 2.
\end{align*}
These facts will be helpful later.

When \(p_t = 1\), there is 0 chance of moving backwards, hence \[
\mathbb{P}(p_{t + 1} - p_t = i) 
 = \left( \frac{1}{2} \right)^{i + 1} \mathbb{I}(i \in \{0, 1, 2, \ldots \}).
\] That is, \([p_{t + 1} - p_t| p_t = 1] \sim G - 1\), or
\([p_{t + 1} - p_t | p_t = 1] \sim G\). Hence \[
\mathbb{E}[p_{t + 1} - p_t | p_t = 1] = \mathbb{E}[G - 1] = \frac{1}{1 / 2} - 1 = 1,
\] and \[
\mathbb{E}[p_{t + 1}^2 - p_t^2 | p_t = 1] =
\mathbb{E}[G^2 - 1] = \frac{2 - (1 / 2)}{(1 / 2)^2} - 1 = 5.
\]

Putting this together, if \[
\phi(x) = x^2 - 3 x + 2,
\] it holds that \[
\mathbb{E}[\phi(p_{t + 1}) - \phi(p_{t}) | p_t > 1] = 1(2) - 3(0) = 2,
\] and \[
\mathbb{E}[\phi(p_{t + 1}) - \phi(p_{t}) | p_t = 1] = 1(5) - 3(1) = 2,
\] and \(\phi(1) = 0\).
\end{proof}

\begin{lemma} The expected number of steps needed for a \(*\) symbol to
reach the right hand side in the deterministic scan chain is bounded
above by \((n^2 - n + 2) / 2\).
\end{lemma}

\begin{proof} Wald's Equation~\cite{wald1944}
gives that \[
\mathbb{E}[\tau] = \mathbb{E}[\phi(p_{\tau})] / 2.
\]
Finding \( \mean[\phi(p_{\tau})] \) is somewhat tricky, since the final value of \( p_{\tau} \) is at least \( n \) rather than exactly \( n \).

When \(\tau = t\), that means that the step moved to the right. The
number of positions moved to the right is a geometric random variable,
which is memoryless over the integers. That is, the number of positions
to the right of \(n\) will also be a geometric random variable with
parameter \(1 / 2\) minus 1.

Hence the expected value of \(\phi(p_\tau)\) will be \begin{align*}
\mathbb{E}[\phi(p_{\tau})] 
  &= \sum_{i = 0}^\infty \left( \frac{1}{2} \right)^{i} ((n + i)^2 - 3(n + i) + 2) \\
  &= n^2 + 2n \mathbb{E}(G - 1) + \mathbb{E}((G - 1)^2) - 3 [n + \mathbb{E}(G - 1)] + 2 \\
  &= n^2 - n + 2.
\end{align*}
\end{proof}

Compare to the random scan chain, which requires on average at most
\(n^2(n - 1)\) steps. Now, each step in the deterministic chain takes
\(n\) times the number of operations for a random scan step, but still
there is a factor of two improvement in the run time by using the
deterministic scan.

\begin{lemma} The probability that \(\tau \geq (n^2 - n + 3) / 2\) is
at most \(1 / 2\).
\end{lemma}

\begin{proof} Since
\(\mathbb{P}(\tau \geq i + 1) \leq \mathbb{P}(\tau \geq i)\) for all
\(i\), the version of Markov's inequality from~\cite{huber2021tail} (Theorem 1) applies which says \[
\mathbb{P}(\tau \geq a) \leq \mathbb{E}[\tau] / (2a - 1).
\] Using \(\mathbb{E}[\tau] \leq (n^2 - n + 2) / 2\) and
\(a = (n^2 - n + 3) / 2\) completes the proof.
\end{proof}

\begin{lemma} After \((n^2 - n + 3) \lceil \log_2(n) \rceil / 2\)
steps, there is at most a \(1 / 2\) chance that any \(*\) symbols remain
in the bounding chain.
\end{lemma}

\begin{proof} Note that
\(\mathbb{P}(\tau \geq a + b) \leq \mathbb{P}(\tau \geq a)\mathbb{P}(\tau \geq b)\).
Therefore, repeating the \(n^2 - n + 3\) steps
\(\lceil \log_2(n) \rceil\) times that for any particular symbol makes
chance of any \(*\) symbols remaining at most \(1 / [2n]\). Using the
union bound then indicates that there is at most a \(1 / 2\) chance that
any of the \(*\) symbols remain after \((n^2 - n + 3)\lceil \log_2(n) \rceil / 2 \) steps.
\end{proof}

Noting that \( n(n^2 - n + 3)\lceil \log_2(n) \rceil \leq 1.83 n^3 \ln(n) \) then completes the proof of Theorem 1.

\hypertarget{implementation}{%
\section{Implementation}\label{implementation}}

The following code implements a variant of
this algorithm in Julia.  Instead of using the same value of \( t \) at each step, at each recursive step the value of \( t \) is doubled.  This allows us to run in time close to the actual value of \( t \) needed before the probability of success is \( 1 / 2 \), rather than the theoretical upper bound.

Here the partial
order is encoded using a 0-1 matrix \texttt{parord} when entry
\((i, j)\) is 1 if \(i \preceq j\) and 0 otherwise.

The following encodes a partial order on \(\{1, 2, 3, 4, 5\}\) where
\(1 \preceq 3\), \(1 \preceq 5\), \(2 \preceq 3\), \( 2 \preceq 5 \),
\(3 \preceq 5\), and \(4 \preceq 5\) together with the symmetric
relations.

\begin{Shaded}
\begin{Highlighting}[]
\NormalTok{parord }\OperatorTok{=}\NormalTok{ [}
    \FloatTok{1} \FloatTok{0} \FloatTok{1} \FloatTok{0} \FloatTok{1}\NormalTok{;}
    \FloatTok{0} \FloatTok{1} \FloatTok{1} \FloatTok{0} \FloatTok{1}\NormalTok{;}
    \FloatTok{0} \FloatTok{0} \FloatTok{1} \FloatTok{0} \FloatTok{1}\NormalTok{;}
    \FloatTok{0} \FloatTok{0} \FloatTok{0} \FloatTok{1} \FloatTok{1}\NormalTok{;}
    \FloatTok{0} \FloatTok{0} \FloatTok{0} \FloatTok{0} \FloatTok{1}
\NormalTok{]}
\end{Highlighting}
\end{Shaded}

\begin{verbatim}
## 5×5 Matrix{Int64}:
##  1  0  1  0  1
##  0  1  1  0  1
##  0  0  1  0  1
##  0  0  0  1  1
##  0  0  0  0  1
\end{verbatim}

Now for the underlying transposition chain.

\begin{Shaded}
\begin{Highlighting}[]
\KeywordTok{function} \FunctionTok{at\_step}\NormalTok{(sigma, coins, parord)}
\NormalTok{    a }\OperatorTok{=} \FunctionTok{length}\NormalTok{(coins)}
    \ControlFlowTok{for}\NormalTok{ i }\KeywordTok{in} \FloatTok{1}\OperatorTok{:}\NormalTok{a}
        \ControlFlowTok{if}\NormalTok{ coins[i] }\OperatorTok{==} \FloatTok{1} \OperatorTok{\&\&}\NormalTok{ parord[sigma[i], sigma[i}\OperatorTok{+}\FloatTok{1}\NormalTok{]] }\OperatorTok{==} \FloatTok{0}
\NormalTok{            sigma[i], sigma[i}\OperatorTok{+}\FloatTok{1}\NormalTok{] }\OperatorTok{=}\NormalTok{ sigma[i}\OperatorTok{+}\FloatTok{1}\NormalTok{], sigma[i]}
        \ControlFlowTok{end}
    \ControlFlowTok{end}
    \ControlFlowTok{return}\NormalTok{ sigma}
\KeywordTok{end}
\end{Highlighting}
\end{Shaded}

\begin{verbatim}
## at_step (generic function with 1 method)
\end{verbatim}

The following is the bounding chain implementation.

\begin{Shaded}
\begin{Highlighting}[]
\KeywordTok{function} \FunctionTok{bc\_step}\NormalTok{(y, coins, parord_ext)}
\NormalTok{    r, k }\OperatorTok{=}\NormalTok{ y}
\NormalTok{    a }\OperatorTok{=} \FunctionTok{length}\NormalTok{(coins)}
    \ControlFlowTok{for}\NormalTok{ i }\KeywordTok{in} \FloatTok{1}\OperatorTok{:}\NormalTok{a}
        \ControlFlowTok{if}\NormalTok{ parord\_ext[r[i], r[i}\OperatorTok{+}\FloatTok{1}\NormalTok{]] }\OperatorTok{==} \FloatTok{0} \OperatorTok{\&\&}\NormalTok{ coins[i] }\OperatorTok{==} \FloatTok{1}
\NormalTok{            r[i], r[i}\OperatorTok{+}\FloatTok{1}\NormalTok{] }\OperatorTok{=}\NormalTok{ r[i}\OperatorTok{+}\FloatTok{1}\NormalTok{], r[i]}
        \ControlFlowTok{end}
        \ControlFlowTok{if}\NormalTok{ r[n] }\OperatorTok{==}\NormalTok{ n }\OperatorTok{+} \FloatTok{1}
\NormalTok{            r[n] }\OperatorTok{=}\NormalTok{ k }\OperatorTok{+} \FloatTok{1}
\NormalTok{            k }\OperatorTok{+=} \FloatTok{1}
        \ControlFlowTok{end}
    \ControlFlowTok{end}
    \ControlFlowTok{return}\NormalTok{ (r, k)}
\KeywordTok{end}
\end{Highlighting}
\end{Shaded}

\begin{verbatim}
## bc_step (generic function with 1 method)
\end{verbatim}

Now for the combined step.

\begin{Shaded}
\begin{Highlighting}[]
\KeywordTok{function} \FunctionTok{sim\_step}\NormalTok{(state, coins, parord)}
\NormalTok{    n }\OperatorTok{=} \FunctionTok{size}\NormalTok{(parord)[}\FloatTok{1}\NormalTok{]}
\NormalTok{    parord\_ext }\OperatorTok{=}\NormalTok{ [parord }\FunctionTok{zeros}\NormalTok{(}\DataTypeTok{Int}\NormalTok{, n, }\FloatTok{1}\NormalTok{); }\FunctionTok{zeros}\NormalTok{(}\DataTypeTok{Int}\NormalTok{, }\FloatTok{1}\NormalTok{, n) }\FloatTok{1}\NormalTok{]}

\NormalTok{    sigma, (r, k) }\OperatorTok{=}\NormalTok{ state}
\NormalTok{    a }\OperatorTok{=} \FunctionTok{length}\NormalTok{(coins)}
    \ControlFlowTok{for}\NormalTok{ i }\KeywordTok{in} \FloatTok{1}\OperatorTok{:}\NormalTok{a}
\NormalTok{        c }\OperatorTok{=}\NormalTok{ coins[i]}
\NormalTok{        c\_prime }\OperatorTok{=}\NormalTok{ sigma[i] }\OperatorTok{==}\NormalTok{ r[i}\OperatorTok{+}\FloatTok{1}\NormalTok{] ? }\FloatTok{1} \OperatorTok{{-}}\NormalTok{ c }\OperatorTok{:}\NormalTok{ c}
        \ControlFlowTok{if}\NormalTok{ parord\_ext[r[i], r[i}\OperatorTok{+}\FloatTok{1}\NormalTok{]] }\OperatorTok{==} \FloatTok{0} \OperatorTok{\&\&}\NormalTok{ c }\OperatorTok{==} \FloatTok{1}
\NormalTok{            r[i], r[i}\OperatorTok{+}\FloatTok{1}\NormalTok{] }\OperatorTok{=}\NormalTok{ r[i}\OperatorTok{+}\FloatTok{1}\NormalTok{], r[i]}
        \ControlFlowTok{end}
        \ControlFlowTok{if}\NormalTok{ r[n] }\OperatorTok{==}\NormalTok{ n }\OperatorTok{+} \FloatTok{1}
\NormalTok{            r[n] }\OperatorTok{=}\NormalTok{ k }\OperatorTok{+} \FloatTok{1}
\NormalTok{            k }\OperatorTok{+=} \FloatTok{1}
        \ControlFlowTok{end}
        \ControlFlowTok{if}\NormalTok{ c\_prime }\OperatorTok{==} \FloatTok{1} \OperatorTok{\&\&}\NormalTok{ parord[sigma[i], sigma[i}\OperatorTok{+}\FloatTok{1}\NormalTok{]] }\OperatorTok{==} \FloatTok{0}
\NormalTok{            sigma[i], sigma[i}\OperatorTok{+}\FloatTok{1}\NormalTok{] }\OperatorTok{=}\NormalTok{ sigma[i}\OperatorTok{+}\FloatTok{1}\NormalTok{], sigma[i]}
        \ControlFlowTok{end}
    \ControlFlowTok{end}
    \ControlFlowTok{return}\NormalTok{ (sigma, (r, k))}
\KeywordTok{end}
\end{Highlighting}
\end{Shaded}

\begin{verbatim}
## sim_step (generic function with 1 method)
\end{verbatim}

Finally, the complete algorithm is as follows.

\begin{Shaded}
\begin{Highlighting}[]
\ImportTok{using} \BuiltInTok{Random}

\KeywordTok{function} \FunctionTok{CFTP\_linear\_extensions}\NormalTok{(t, parord)}
\NormalTok{    n }\OperatorTok{=} \FunctionTok{size}\NormalTok{(parord)[}\FloatTok{1}\NormalTok{]}
\NormalTok{    all\_coins }\OperatorTok{=} \FunctionTok{bitrand}\NormalTok{(n }\OperatorTok{{-}} \FloatTok{1}\NormalTok{, t)}

\NormalTok{    sigma }\OperatorTok{=} \FunctionTok{collect}\NormalTok{(}\FloatTok{1}\OperatorTok{:}\NormalTok{n)}
\NormalTok{    y }\OperatorTok{=}\NormalTok{ (}\FunctionTok{vcat}\NormalTok{(}\FunctionTok{fill}\NormalTok{(n }\OperatorTok{+} \FloatTok{1}\NormalTok{, n }\OperatorTok{{-}} \FloatTok{1}\NormalTok{), }\FloatTok{1}\NormalTok{), }\FloatTok{1}\NormalTok{)}
\NormalTok{    state }\OperatorTok{=}\NormalTok{ (sigma, y)}

    \ControlFlowTok{for}\NormalTok{ t\_prime }\KeywordTok{in} \FloatTok{1}\OperatorTok{:}\NormalTok{t}
\NormalTok{        state }\OperatorTok{=} \FunctionTok{sim\_step}\NormalTok{(state, all\_coins[}\OperatorTok{:}\NormalTok{, t\_prime], parord)}
    \ControlFlowTok{end}

    \ControlFlowTok{if}\NormalTok{ state[}\FloatTok{2}\NormalTok{][}\FloatTok{2}\NormalTok{] }\OperatorTok{==}\NormalTok{ n}
        \ControlFlowTok{return}\NormalTok{ state[}\FloatTok{1}\NormalTok{]}
    \ControlFlowTok{else}
\NormalTok{        state }\OperatorTok{=}\NormalTok{ (}\FunctionTok{CFTP\_linear\_extensions}\NormalTok{(}\FloatTok{2}\NormalTok{ * t, parord), }
\NormalTok{                 (}\FunctionTok{vcat}\NormalTok{(}\FunctionTok{fill}\NormalTok{(n}\OperatorTok{+}\FloatTok{1}\NormalTok{, n}\OperatorTok{{-}}\FloatTok{1}\NormalTok{), }\FloatTok{1}\NormalTok{), }\FloatTok{1}\NormalTok{))}
        \ControlFlowTok{for}\NormalTok{ t\_prime }\KeywordTok{in} \FloatTok{1}\OperatorTok{:}\NormalTok{t}
\NormalTok{            state }\OperatorTok{=} \FunctionTok{sim\_step}\NormalTok{(state, all\_coins[}\OperatorTok{:}\NormalTok{, t\_prime], parord)}
        \ControlFlowTok{end}
    \ControlFlowTok{end}
    \ControlFlowTok{return}\NormalTok{ state[}\FloatTok{1}\NormalTok{]}
\KeywordTok{end}
\end{Highlighting}
\end{Shaded}

\begin{verbatim}
## CFTP_linear_extensions (generic function with 1 method)
\end{verbatim}

The algorithm now produces uniform draws from the partial order,
regardless of the initial number of steps given.

\begin{Shaded}
\begin{Highlighting}[]
\FunctionTok{CFTP\_linear\_extensions}\NormalTok{(}\FloatTok{1}\NormalTok{, parord)}
\end{Highlighting}
\end{Shaded}

\begin{verbatim}
## 5-element Vector{Int64}:
##  2
##  1
##  3
##  4
##  5
\end{verbatim}

\begin{Shaded}
\begin{Highlighting}[]
\FunctionTok{CFTP\_linear\_extensions}\NormalTok{(}\FloatTok{5}\NormalTok{, parord)}
\end{Highlighting}
\end{Shaded}

\begin{verbatim}
## 5-element Vector{Int64}:
##  4
##  2
##  1
##  3
##  5
\end{verbatim}

\begin{Shaded}
\begin{Highlighting}[]
\FunctionTok{CFTP\_linear\_extensions}\NormalTok{(}\FloatTok{10}\NormalTok{, parord)}
\end{Highlighting}
\end{Shaded}

\begin{verbatim}
## 5-element Vector{Int64}:
##  4
##  1
##  2
##  3
##  5
\end{verbatim}

To test the algorithm on the partial order, assign each of the 8
possible outputs an integer from 1 to 8.

\begin{Shaded}
\begin{Highlighting}[]
\KeywordTok{function} \FunctionTok{assign}\NormalTok{(parord)}
\NormalTok{    sigma }\OperatorTok{=} \FunctionTok{CFTP\_linear\_extensions}\NormalTok{(}\FloatTok{10}\NormalTok{, parord)}
\NormalTok{    output }\OperatorTok{=} \FunctionTok{zeros}\NormalTok{(}\DataTypeTok{Int}\NormalTok{, }\FloatTok{8}\NormalTok{)}

    \ControlFlowTok{if}\NormalTok{ sigma }\OperatorTok{==}\NormalTok{ [}\FloatTok{1}\NormalTok{, }\FloatTok{2}\NormalTok{, }\FloatTok{3}\NormalTok{, }\FloatTok{4}\NormalTok{, }\FloatTok{5}\NormalTok{]}
\NormalTok{        output[}\FloatTok{1}\NormalTok{] }\OperatorTok{=} \FloatTok{1}
    \ControlFlowTok{elseif}\NormalTok{ sigma }\OperatorTok{==}\NormalTok{ [}\FloatTok{1}\NormalTok{, }\FloatTok{2}\NormalTok{, }\FloatTok{4}\NormalTok{, }\FloatTok{3}\NormalTok{, }\FloatTok{5}\NormalTok{]}
\NormalTok{        output[}\FloatTok{2}\NormalTok{] }\OperatorTok{=} \FloatTok{1}
    \ControlFlowTok{elseif}\NormalTok{ sigma }\OperatorTok{==}\NormalTok{ [}\FloatTok{1}\NormalTok{, }\FloatTok{4}\NormalTok{, }\FloatTok{2}\NormalTok{, }\FloatTok{3}\NormalTok{, }\FloatTok{5}\NormalTok{]}
\NormalTok{        output[}\FloatTok{3}\NormalTok{] }\OperatorTok{=} \FloatTok{1}
    \ControlFlowTok{elseif}\NormalTok{ sigma }\OperatorTok{==}\NormalTok{ [}\FloatTok{4}\NormalTok{, }\FloatTok{1}\NormalTok{, }\FloatTok{2}\NormalTok{, }\FloatTok{3}\NormalTok{, }\FloatTok{5}\NormalTok{]}
\NormalTok{        output[}\FloatTok{4}\NormalTok{] }\OperatorTok{=} \FloatTok{1}
    \ControlFlowTok{elseif}\NormalTok{ sigma }\OperatorTok{==}\NormalTok{ [}\FloatTok{2}\NormalTok{, }\FloatTok{1}\NormalTok{, }\FloatTok{3}\NormalTok{, }\FloatTok{4}\NormalTok{, }\FloatTok{5}\NormalTok{]}
\NormalTok{        output[}\FloatTok{5}\NormalTok{] }\OperatorTok{=} \FloatTok{1}
    \ControlFlowTok{elseif}\NormalTok{ sigma }\OperatorTok{==}\NormalTok{ [}\FloatTok{2}\NormalTok{, }\FloatTok{1}\NormalTok{, }\FloatTok{4}\NormalTok{, }\FloatTok{3}\NormalTok{, }\FloatTok{5}\NormalTok{]}
\NormalTok{        output[}\FloatTok{6}\NormalTok{] }\OperatorTok{=} \FloatTok{1}
    \ControlFlowTok{elseif}\NormalTok{ sigma }\OperatorTok{==}\NormalTok{ [}\FloatTok{2}\NormalTok{, }\FloatTok{4}\NormalTok{, }\FloatTok{1}\NormalTok{, }\FloatTok{3}\NormalTok{, }\FloatTok{5}\NormalTok{]}
\NormalTok{        output[}\FloatTok{7}\NormalTok{] }\OperatorTok{=} \FloatTok{1}
    \ControlFlowTok{elseif}\NormalTok{ sigma }\OperatorTok{==}\NormalTok{ [}\FloatTok{4}\NormalTok{, }\FloatTok{2}\NormalTok{, }\FloatTok{1}\NormalTok{, }\FloatTok{3}\NormalTok{, }\FloatTok{5}\NormalTok{]}
\NormalTok{        output[}\FloatTok{8}\NormalTok{] }\OperatorTok{=} \FloatTok{1}
    \ControlFlowTok{end}

    \ControlFlowTok{return}\NormalTok{ output}
\KeywordTok{end}
\end{Highlighting}
\end{Shaded}

\begin{verbatim}
## assign (generic function with 1 method)
\end{verbatim}

Run this multiple times to get an idea of the output

\begin{Shaded}
\begin{Highlighting}[]
\KeywordTok{function} \FunctionTok{test\_run}\NormalTok{(trials, parord)}
\NormalTok{    res }\OperatorTok{=} \FunctionTok{zeros}\NormalTok{(}\DataTypeTok{Int}\NormalTok{, }\FloatTok{8}\NormalTok{)}
    \ControlFlowTok{for}\NormalTok{ \_ }\KeywordTok{in} \FloatTok{1}\OperatorTok{:}\NormalTok{trials}
\NormalTok{        res }\OperatorTok{.+=} \FunctionTok{assign}\NormalTok{(parord)}
    \ControlFlowTok{end}
    \ControlFlowTok{return}\NormalTok{ res}
\KeywordTok{end}
\end{Highlighting}
\end{Shaded}

\begin{verbatim}
## test_run (generic function with 1 method)
\end{verbatim}

To time the runs, add a library.

\begin{Shaded}
\begin{Highlighting}[]
\ImportTok{using} \BuiltInTok{Dates}
\end{Highlighting}
\end{Shaded}

Now see what happens with the run.

\begin{Shaded}
\begin{Highlighting}[]
\NormalTok{trials }\OperatorTok{=} \FloatTok{10000}\NormalTok{;}
\NormalTok{start\_time }\OperatorTok{=} \FunctionTok{now}\NormalTok{();}
\NormalTok{observed }\OperatorTok{=} \FunctionTok{test\_run}\NormalTok{(trials, parord);}
\NormalTok{result }\OperatorTok{=}\NormalTok{ observed }\OperatorTok{/}\NormalTok{ trials;}
\NormalTok{end\_time }\OperatorTok{=} \FunctionTok{now}\NormalTok{()}
\end{Highlighting}
\end{Shaded}

\begin{verbatim}
## 2025-06-14T06:50:59.902
\end{verbatim}

\begin{Shaded}
\begin{Highlighting}[]
\FunctionTok{println}\NormalTok{(}\StringTok{"Number of trials: "}\NormalTok{, trials)}
\end{Highlighting}
\end{Shaded}

\begin{verbatim}
## Number of trials: 10000
\end{verbatim}

\begin{Shaded}
\begin{Highlighting}[]
\FunctionTok{println}\NormalTok{(}\StringTok{"Elapsed time: "}\NormalTok{, end\_time }\OperatorTok{{-}}\NormalTok{ start\_time)}
\end{Highlighting}
\end{Shaded}

\begin{verbatim}
## Elapsed time: 1299 milliseconds
\end{verbatim}

\begin{Shaded}
\begin{Highlighting}[]
\FunctionTok{println}\NormalTok{(}\StringTok{"Result: "}\NormalTok{, result)}
\end{Highlighting}
\end{Shaded}

\begin{verbatim}
## Result: [0.1248, 0.1278, 0.1214, 0.1263, 0.1261, 0.1229, 0.1243, 0.1264]
\end{verbatim}

Now suppose that \(2 \preceq 4\) is added to the partial order.

\begin{Shaded}
\begin{Highlighting}[]
\NormalTok{parord[}\FloatTok{2}\NormalTok{, }\FloatTok{4}\NormalTok{] }\OperatorTok{=} \FloatTok{1}\NormalTok{;}
\NormalTok{parord}
\end{Highlighting}
\end{Shaded}

\begin{verbatim}
## 5×5 Matrix{Int64}:
##  1  0  1  0  1
##  0  1  1  1  1
##  0  0  1  0  1
##  0  0  0  1  1
##  0  0  0  0  1
\end{verbatim}

Then three of the eight partial orders now receive probability 0, and
the remaining five receive a 20\% chance.

\begin{Shaded}
\begin{Highlighting}[]
\NormalTok{trials }\OperatorTok{=} \FloatTok{10000}\NormalTok{;}
\NormalTok{start\_time }\OperatorTok{=} \FunctionTok{now}\NormalTok{();}
\NormalTok{observed }\OperatorTok{=} \FunctionTok{test\_run}\NormalTok{(trials, parord);}
\NormalTok{result }\OperatorTok{=}\NormalTok{ observed }\OperatorTok{/}\NormalTok{ trials;}
\NormalTok{end\_time }\OperatorTok{=} \FunctionTok{now}\NormalTok{()}
\end{Highlighting}
\end{Shaded}

\begin{verbatim}
## 2025-06-14T06:51:01.619
\end{verbatim}

\begin{Shaded}
\begin{Highlighting}[]
\FunctionTok{println}\NormalTok{(}\StringTok{"Number of trials: "}\NormalTok{, trials)}
\end{Highlighting}
\end{Shaded}

\begin{verbatim}
## Number of trials: 10000
\end{verbatim}

\begin{Shaded}
\begin{Highlighting}[]
\FunctionTok{println}\NormalTok{(}\StringTok{"Elapsed time: "}\NormalTok{, end\_time }\OperatorTok{{-}}\NormalTok{ start\_time)}
\end{Highlighting}
\end{Shaded}

\begin{verbatim}
## Elapsed time: 1203 milliseconds
\end{verbatim}

\begin{Shaded}
\begin{Highlighting}[]
\FunctionTok{println}\NormalTok{(}\StringTok{"Result: "}\NormalTok{, result)}
\end{Highlighting}
\end{Shaded}

\begin{verbatim}
## Result: [0.2035, 0.1971, 0.0, 0.0, 0.1943, 0.2075, 0.1976, 0.0]
\end{verbatim}

Note that because the bounding chain is only concerned with eliminating
\(*\) symbols, the distribution of the running time of the two partial
orders to generate a sample will be exactly the same!

\bibliographystyle{plain}


\end{document}